\documentclass[journal]{IEEEtran}

\usepackage{amsthm}
\usepackage{amsmath}
\usepackage{amsfonts}
\usepackage{enumitem}
\usepackage{amssymb}
\usepackage{xcolor}
\usepackage{esvect}
\usepackage{url}
\usepackage{graphicx}
\usepackage{subcaption}
\usepackage[ruled,vlined]{algorithm2e}
\usepackage{xcolor}
\usepackage[normalem]{ulem} 

\graphicspath{ {./images/} }

\newtheorem{assumption}{Assumption}[section]
\newtheorem{corollary}[assumption]{Corollary}
\newtheorem{theorem}[assumption]{Theorem}
\newtheorem{lemma}[assumption]{Lemma}

\newcommand{\SNR}[0]{\text{SNR}}
\newcommand{\E}[0]{\mathbb{E}}

\newcommand{\Cov}[0]{\operatorname{Cov}}
\newcommand{\R}[0]{\mathbb{R}}
\newcommand{\group}[0]{\mathcal{A}}

\title{The generalized method of moments for multi-reference alignment}
\begin{document}

\author{Asaf Abas,
	Tamir Bendory,
	and~Nir Sharon
	\thanks{
		A.\ Abas and N.\ Sharon are with the department of Applied Mathematics, Tel Aviv University, Israel. T.\ Bendory is with the school of Electrical Engineering, Tel Aviv University, Israel, (e-mail: asafabas@gmail.com; bendory@tauex.tau.ac.il; nir.sharon@math.tau.ac.il).
		The research is partially supported by the NSF-BSF award 2019752. T. Bendory is also supported by the Zimin Institute for Engineering Solutions Advancing Better Lives.  N.~Sharon is also supported by BSF grant no. 2018230. }
} 
\maketitle

\begin{abstract}
This paper studies the application of the generalized method of moments (GMM) to multi-reference alignment (MRA): the problem of estimating a signal from its circularly-translated and noisy copies. We begin by proving that the GMM estimator maintains its asymptotic optimality for statistical models with group symmetry, including MRA. Then, we conduct a comprehensive numerical study and show that the GMM substantially outperforms the classical method of moments, whose application to MRA has been studied thoroughly in the literature. We also formulate the GMM {to estimate} a three-dimensional molecular structure using cryo-electron microscopy and present numerical results on simulated data.	
\end{abstract}

	\maketitle

\section{Introduction}
Multi-reference alignment (MRA) is the problem of recovering a signal $x	\in \R^L$ from
\begin{equation}\label{Eq-1-1}
y_i = R_{s_i}  x + \varepsilon_i  , \quad i=1,\ldots,N,
\end{equation}
where~$R_{s}$ circularly-translates {the signal} by~$s$ elements, i.e., $(R_s x)[j] = x\left[(j-s) \bmod L\right]$, and $\varepsilon_i \overset{i.i.d.}{\sim} \mathcal{N}(0,\Sigma)$. The translations~$s_i$ are sampled from an unknown distribution~$\rho$. 
Since the translations {$R_{s_1},\ldots,R_{s_N}$} are unknown, we cannot distinguish {between~\eqref{Eq-1-1} and the set of observations $y_i =R_{s_i - \tilde s}(R_{\tilde s} x)+\varepsilon_i$ for any $\tilde s\in\mathbb{Z}$}. Thus, our goal is to estimate~$x$, up to {a} circular translation. In group theory terminology, the set of signals $\{R_sx\}_{s=0}^{L-1}$ is called the \textit{orbit} of the signal under the group of circular translations.

The MRA model is motivated by applications in signal processing~\cite{zwart2003fast}	 and structural biology \cite{scheres2005maximum, theobald2012optimal, ma2019heterogeneous,Bendory2018a}. In particular, it has been demonstrated as a useful mathematical abstraction for single-particle cryo-electron microscopy (cryo-EM), an emerging technology to elucidate the 3-D structure of biomolecules \cite{Frank2006, Nogales2015,Vinothkumar2016,Bendory2020, Singer2020}. Cryo-EM is the chief motivation of this paper, and is the main focus of Section~\ref{Sc-4}. 

If the signal-to-noise ratio (SNR) is sufficiently high, one can recover the signal~$x$ by estimating the circular translations~{$R_{s_1},\ldots,R_{s_N}$} using a variety of synchronization  algorithms~\cite{Singer2011, Bandeira2014multireference, Boumal2016, perry2018message}, align them (i.e., undo the circular translations), and average out the noise. However, low SNR hinders reliable shift estimation~\cite{aguerrebere2016fundamental,bendory2019multi}, and thus in this regime one must estimate the signal directly.

The \textit{method of moments} (MoM) is a classical statistical inference technique, tracing back to 1894~\cite{pearson1894contributions}. The MoM estimator, which is described in detail in Section~\ref{Sc-2-1}, results in a set of parameters whose moments agree with the empirical moments of the observed data~$\{y_i\}_{i=1}^N$. {A single-pass through the observations is required to compute the empirical moments.} This {single-pass requirement} stands in contrast to main-stream parameter estimation techniques, such as  maximum likelihood {estimation}, which usually {iterates} over the data. Thus, the MoM is an attractive computational framework for massive data sets. In addition, it was shown that in the low SNR regime, when $\SNR \to 0$, $N \to \infty$, and~$L$ is fixed, the MoM achieves the optimal sample complexity {of MRA, namely, }the minimal number of observations needed for recovery up to {an} arbitrary precision~\cite{bandeira2017estimation,bandeira2017optimal,abbe2018estimation,Abbe2017, perry2019sample,romanov2020multi}. {Notwithstanding}, finding parameters that approximate the observable moments, namely moment fitting, {frequently} requires solving a system of nonlinear polynomial equations, a challenging computational task for high-dimensional data.

A standard method for fitting the analytic and the empirical moments is by minimizing a least-squares (LS) objective, which often results in {a sub-optimal solution}. As a remedy, this paper studies the \textit{generalized method of moments} (GMM), which suggests matching the moments by a weighted LS objective~\cite{Hansen1982}, while providing an explicit expression of the optimal set of weights. Choosing the optimal weights enjoys appealing statistical properties, as discussed in Section~\ref{Sc-2.5}. 

This work studies the GMM and its application to the MRA problem. Applying the GMM to MRA raises two main challenges: the dimensionality of the problem, which is typically high, and the solution's symmetry (i.e., the solution is defined up to a symmetry). In particular, classical GMM theory shows that the GMM provides {an} optimal estimator if a single set of parameters fits the observable moments. Namely, the polynomial system of equations has a unique solution. Unfortunately, this is never the case for MRA as the solution is defined up to a circular translation. Filling a theoretical gap, this work proves that the GMM retains its favorable 
 statistical properties even when the statistical model has an intrinsic symmetry. 

The paper is organized as follows. In Section~\ref{Sec-2}, we describe in detail the MoM and the GMM estimators. Next, Section~\ref{Sc-2.5} {discusses} the properties of the GMM estimator and our theoretical work to extend it for models {with} intrinsic symmetry. Section~\ref{Sc-3} provides a comprehensive numerical study. Our study demonstrates that the GMM outperforms the MoM in a variety of noise models and levels. We also provide a heuristic to predict when the performance gap between the GMM and the MoM is expected to be significant. Ultimately, Section~\ref{Sc-4} formulates the GMM for the problem of recovering a 3-D molecular structure using cryo-EM and presents initial numerical results. 

\section{The Generalized Method of Moments} \label{Sec-2}

\subsection{The method of moments}\label{Sc-2-1}
Before introducing the GMM framework, we begin by presenting the classical MoM. Suppose that a random variable $y\in\R^r$ is drawn from a distribution which can be characterized by a set of parameters~$\theta_0\in\Theta$, where~$\Theta$ is a compact space. The goal is to recover the unknown parameters~$\theta_0$ from~$N$ samples $y_1,\ldots, y_N$. In the MoM, the underlying idea is to estimate~$\theta_0$ from the first~$k_M$ empirical moments of the observations. We calculate the empirical moments from the data by averaging over the moments of individual observations. In particular, the~$k$-th empirical moment is {defined as}
\begin{equation}\label{Eq-2-6}
\hat{M}_k = \frac{1}{N}\sum_{i=1}^N  y_i^{\otimes k},
\end{equation}
where~$y_i^{\otimes k}$ is a tensor with~$r^k$ entries. Each entry is  given by $y_i^{\otimes k}(n)=\prod_{j=1}^k y_i(n_j)$, where $n = (n_1, \ldots, n_k)$ such that $1\le n_i \le r$ for $i = 1,\ldots,k$.
By the law of large numbers, for a large~$N$ we have,
\begin{equation*}
	\hat{M}_k \approx M_k(\theta_0) := \E\left[y^{\otimes k} \right], \quad k=1,\ldots,k_M, 
\end{equation*}
where $\E[\cdot]$ denotes expectation. 

The MoM consists of two stages. First, one computes the first $k_M$ observable moments  $\hat M_k$ for $k=1,\ldots,k_M$ from the data. 
In this work, we {usually} use the first two moments, that is, $k_M = 2$.
In the second stage, we wish to find a set of parameters~$\theta$ so that $\hat{M}_k \approx M_k(\theta)$ for $k=1,\ldots,k_M$; this {occasionally} entails solving a system of polynomial equations. When a closed-form solution is not available, it is common to minimize a LS objective function,
\begin{equation}\label{Eq-2-2}
\hat{\theta}^{LS}_N =\arg\min_{\theta \in \Theta} \sum_{k=1}^{k_M} || M_k(\theta) - \hat{M}_k||_\text{F}^2,
\end{equation}
where $\hat{\theta}^{LS}_N$ denotes the LS estimator, and~$\Theta$ is the parameter space. 

{We note that the MoM is computationally attractive only when the number of observations~$N$ is much larger than the signal's length~$L$. This is true since the computational complexity of the MoM is proportional to~$L^{k_M}N$, whereas the complexity of methods that maximize the likelihood function (such as expectation-maximization) usually scales as~$TLN$, where~$T$ is the number of  passes through the data (which tends to increase as the SNR decreases~\cite{janco2021accelerated,Bendory2018}).}

\subsection{The GMM framework}\label{Sc-2-2}

In its most simplified form, the GMM generalizes~\eqref{Eq-2-2} by replacing the LS objective with a weighted LS. In particular, a specific choice of weights guarantees favorable asymptotic statistical properties, such as {the} minimal asymptotic variance of the estimation error. We introduce these properties in detail in Section~\ref{Sc-2.5}. 

Let us define the \textit{moment function}, $f(\theta, y)\colon \Theta \times \R^r \to \R^q$. The moment function needs to be chosen such that its expectation value is zero only at a single point $\theta=\theta_0$. Namely,
\begin{equation}\label{Eq-2-4}
\E\left[f(\theta,y)\right] = 0 \quad \text{if and only if} \quad \theta = \theta_0.
\end{equation}
We refer to~\eqref{Eq-2-4} as the \textit{uniqueness of the parameter set} condition.
Henceforth, we choose the moment function to be
\begin{equation}\label{Eq-2-3}
f(\theta, y_i) = \left[M_1(\theta)-y_i; \ldots; M_{k_M}(\theta) - y_i^{\otimes {k_M}}\right].
\end{equation}
For convenience, we treat each moment as a column vector and the right hand side of~\eqref{Eq-2-3} as their concatenation. For example  $M_2(\cdot)\in\R^{r^2}$, {and} $f(\cdot)$ is in~$\R^{r+r^2+...+r^{k_M}}$. While we choose~$f$ as in~\eqref{Eq-2-3}, {any} moment function can be  chosen as long as it satisfies the uniqueness condition and a few additional regularity conditions ({as} introduced in Appendix~\ref{Ap-1}). This flexibility enables the GMM to be applied to a wide range of problems, such as subspace estimation~\cite{Fan2018}.

The estimated sample moment function is the average of~$f$ over~$N$ observations:
\begin{equation}\label{Eq-2-5}
g_N(\theta) = \frac{1}{N} \sum_{i=1}^N f(\theta, y_i).
\end{equation}
The GMM estimator is defined as the minimizer of the weighted LS expression, 	
\begin{equation} \label{Eq-2-1}
\hat{\theta}_N = \arg\min_{\theta \in \Theta} \ g_N(\theta)^T W_N g_N(\theta).
\end{equation}
Here, $W_N$ is a fixed positive semi-definite (PSD) matrix. Note that the LS estimator~\eqref{Eq-2-2} is a special case of~\eqref{Eq-2-1}  when~$f$ is chosen as {in}~\eqref{Eq-2-3} and $W_N=I$. 

\section{Large Sample Properties}\label{Sc-2.5}

Before presenting the statistical properties of the GMM, we fix notation. We denote by $\overset{p}{\to}$ and $\overset{d}{\to}$ convergence in  probability and in distribution, respectively. Let 
\begin{equation} \label{eqn:cov_mat_S}
S := \lim_{N\to \infty}\Cov\left[\sqrt{N}g_N(\theta_0)\right],
\end{equation}
be the covariance matrix of the estimated sample moment function~\eqref{Eq-2-5} at the ground truth~$\theta_0$. We denote by $\{W_N\}_{N=1}^\infty$ a sequence of PSD matrices which converges almost surely to a positive definite (PD) matrix~$W$. The expectation of the Jacobian of the moment function at the ground truth~$\theta_0$ is denoted by $G_0 = \E\left[\partial f(\theta_0, y) / \partial \theta^T\right]$. 

\subsection{GMM with a unique set of parameters}


{The large sample properties of the GMM estimator, under the uniqueness of the parameter set~\eqref{Eq-2-4},   were derived in~\cite{Hansen1982}, and are presented in the following theorem.}

\begin{theorem}\label{Thm-2-6}
	Under the uniqueness of the parameter set~\eqref{Eq-2-4}, and {the} regularity conditions {\ref{As-Ap1-1*}-\ref{As-Ap1-8*} of Appendix~\ref{Ap-1}}, the GMM estimator satisfies:
	\begin{enumerate}[label={\Alph*}.]
		\item  \label{Thm-2-2}
		\textnormal{(Consistency)} $\hat{\theta}_N \overset{p}{\to} \theta_0$.
		
		\item \label{Thm-2-3} \textnormal{(Asymptotic normality)}
		\[\sqrt{N} ( \hat{\theta}_N - \theta_0) \overset{d}{\to} \mathcal{N}(0, M S M^T ),\] where $M =[G_0^T W  G_0]^{-1} G_0^T  W$.
		
		\item \label{Thm-2-5} \textnormal{(Optimal choice of a weighting matrix)} The minimum asymptotic variance of $\hat{\theta}_N$ is given by $(G_0^T S^{-1} G_0)^{-1}$ and is attained by $W = S^{-1}$.
	\end{enumerate}
\end{theorem}
Theorem~\ref{Thm-2-6} provides a matrix~$W$ {that} guarantees a minimal asymptotic variance of the estimator’s error. In Appendix~\ref{Ap-1}, we present the regularity conditions of Theorem~\ref{Thm-2-6} in detail. {While most of these conditions hold for the MRA model, e.g., continuity of~$f(\cdot)$, the} intrinsic symmetry of the MRA model~\eqref{Eq-1-1} violates the uniqueness condition~\eqref{Eq-2-4}. In {the next section}, we extend Theorem~\ref{Thm-2-6} and prove that the optimality of the GMM remains true even if there is a unique orbit of signals that fits the moments, rather than a unique signal.

The covariance matrix $S$ {of~\eqref{eqn:cov_mat_S}, which plays a central role in Theorem~\ref{Thm-2-6},} is required to be a PD matrix, see Appendix~\ref{Ap-1}. Therefore, the moment function must be chosen so that~$S$ is full-rank; see for example~\cite{Doran2006}.
In this work, we noticed that if we remove the repeating  	entries of~$f$ (that appear due to the inherent symmetries of the moments), the covariance is indeed full rank (although, in some cases, ill-conditioned).


It is important to note that in practice, the ground truth~$\theta_0$ is unknown a priori, so we cannot use the optimal weighting matrix. A common heuristic is to replace~\eqref{Eq-2-1} with an iterative scheme {called iterative GMM}. {However, for our specific choice of  moment function~$f$~\eqref{Eq-2-3}, the covariance of $g_N$ depends solely on the observations~$\{y_i\}_{i=1}^N$, and not on the parameter set~$\theta$, namely,
\begin{equation}\label{Eq-2-7}
	\Cov[g_N(\theta)] = \Cov\left[y_i; \ldots; y_i^{\otimes {k_M}}\right].
\end{equation}
Therefore, one can compute  both the moment function and the covariance matrix~$S$~\eqref{eqn:cov_mat_S} in a single pass. The algorithm is detailed in Algorithm~\ref{Alg-2-1}.} 

\begin{algorithm}
	\SetAlgoLined
	\KwOut{{The estimated parameter vector $\hat{\theta}_{\text{GMM}} $ }}
	\textbf{Input:}  A set of observations $\{y_i\}_{i=1}^N$ 
	\begin{enumerate}
		\item Compute the {estimated sample} moment function $g_N$, given in~\eqref{Eq-2-5}
		\item Compute the covariance $S = \Cov\left[y_i; \ldots; y_i^{\otimes {k_M}}\right]$ and its inverse $W = S^{-1}$ 
		\item Solve
		$\arg\min_{\theta \in \Theta} g_N(\theta)^T W g_N(\theta)$
	\end{enumerate}
	\caption{{The GMM estimator}}
	\label{Alg-2-1}
\end{algorithm}


\subsection{GMM with a unique orbit}\label{Sc-2.5-1}
We are now ready to introduce the main theoretical contribution of this paper: extending Theorem~\ref{Thm-2-6} to the case where there is an orbit of signals that agrees with the observable moments. This is the case, for example, in the MRA model~\eqref{Eq-1-1} {as well as in cryo-EM~\cite{Bendory2020}}. 

Let~$\group$ be a group acting on a vector space~$\Theta$. We denote the group action by $a\circ\theta$. To extend Theorem~\ref{Thm-2-6}, we make {the following} two assumptions.
\begin{assumption} \label{As-2-2}
	Global identification up to an orbit: 
	\[\forall a \in \group: \  \E \left[ f(\theta,y)\right] = 0 \text{ if and only if } \theta = a \circ \theta_0.\]
\end{assumption}
\begin{assumption}{}\label{As-2-1}
	Symmetry of the moment function:
	\[\forall a \in \group, \forall \theta \in \Theta, \forall i \in \{1,..., N\} : \  f(a \circ \theta, y_i) = f(\theta,y_i).\]
\end{assumption}

The next {theorem} shows that under {the above} assumptions, the large sample properties of Theorem~\ref{Thm-2-6} remain true for problems with symmetry. In particular, Assumptions~\ref{As-2-2} and~\ref{As-2-1} are sufficient to guarantee consistency, as shown in Theorem~\ref{Thm-2-1}; the rest of the properties are direct corollaries. The proofs {of Theorem~\ref{Thm-2-1} and Corollary~\ref{Thm-2-4}} are provided in Appendices~\ref{Ap-2} and~\ref{Ap-3}{, respectively}.

\begin{theorem}[Consistency] \label{Thm-2-1}
	Under Assumptions \ref{As-2-2}, \ref{As-2-1}, and {the regularity conditions \ref{As-Ap1-1*}-\ref{As-Ap1-5*} of Appendix~\ref{Ap-1}}, there exists a sequence $\{a_{N}\}_{N=1}^\infty \subset \group$, such that:   
	\[a_N \circ \hat{\theta}_N \overset{p}{\to} \theta_0,\]
	where~$\hat\theta_N$ is the GMM estimator.
\end{theorem}

\begin{corollary}\label{Thm-2-4}
	Under Assumptions \ref{As-2-2}, \ref{As-2-1}, and {the regularity conditions \ref{As-Ap1-1*}-\ref{As-Ap1-8*} of Appendix~\ref{Ap-1}}, there exists a sequence $\{a_{N}\}_{N=1}^\infty \subset \group$, such that:
	\begin{enumerate}[label={\Alph*}.]
		\item \textnormal{(Asymptotic normality)}  
		\[\sqrt{N} ( a_{N} \circ \hat{\theta}_N - \theta_0) \overset{d}{\to} \mathcal{N}(0, M S M^T ),\] where $M =[G_0^T W  G_0]^{-1} G_0^T  W$.
		
		\item \textnormal{(Optimal choice of weighting matrix)} The minimum asymptotic variance of $a_N \circ \hat{\theta}_N$ is given by $(G_0^T S^{-1} G_0)^{-1}$ and is attained for $W = S^{-1}$.
	\end{enumerate}
\end{corollary}

\section{Application of GMM to MRA} \label{Sc-3}

\subsection{Moments}
The first two moments of {the MRA model}~\eqref{Eq-1-1} are given by~\cite{Abbe2017}: 
\begin{equation*}
M_1(x,\rho)  = x * \rho = C_x \rho = C_\rho x,
\end{equation*}
\begin{equation*}
M_2(x,\rho)  = C_x D_\rho C_x^T + \Sigma.
\end{equation*}
Here, $*$ denotes a convolution, $C_x$ is a circulant matrix whose first column is~$x$, $D_\rho$ is a diagonal matrix whose diagonal consists {of} the entries of~$\rho$, and~$\Sigma$ is the covariance matrix of the noise. These moments can be estimated from the observations by the empirical moments $\hat{M}_1$ and $\hat{M}_2$ as in~\eqref{Eq-2-6}. By the law of large numbers, if~$N$ is large enough then $\hat M_1\approx M_1$ and $\hat M_2\approx M_2$.

In~\cite{Abbe2017}, it was shown that the first two moments suffice to recover the orbit of~$x$ and~$\rho$, for {almost any non-uniform distribution}~$\rho$ {and} if the discrete Fourier transform~(DFT) of~$x$ is non-vanishing. Therefore, a natural candidate for the moment function is 
\begin{equation} \label{Eq-3-1}
f(x,\rho, y_i) = 
\begin{pmatrix}
M_1(x, \rho) - y_i\\
M_2(x, \rho) - y_i y_i^T
\end{pmatrix}.
\end{equation}

Since the second moment is a symmetric matrix, we remove all recurrent entries (e.g., {eliminate} the left lower  triangle of the second moment matrix). Empirically, {this} modification results in a full-rank covariance~$S$. With a slight abuse of notation, we continue using the notation of~\eqref{Eq-3-1} after removing the recurrent entries. Thus, ${f} \in\R^{L+(L+1)L/2}$.

Before diving into the numerical results, we need to verify that the MRA model~\eqref{Eq-1-1} and the moment function~\eqref{Eq-3-1} satisfy Assumptions~\ref{As-2-2} and~\ref{As-2-1}. First, we note that Assumption~\ref{As-2-2} is satisfied since the first two moments determine the orbit of the signal uniquely (under the aforementioned conditions). Let us define a group element {$a_i$} acting on $\theta = [\rho; x]$ by ${a_i}\circ[\rho;x] = [R_{-i}  \rho; R_i  x]$. Since $M_1$ and $M_2$ are invariant under this group action, the  moment function~\eqref{Eq-3-1} is  invariant as well, and thus Assumption~\ref{As-2-1} holds.

\subsection{Experimental setting}
Due to the inherent symmetry of the MRA model, we define the relative error as:
\begin{equation}\label{Eq-3-2}
\operatorname{error}(x, \hat{x}) = \min_{0\le s \le L-1} \frac{\| R_s \hat{x} - x \|_2}{\|x \|_2},
\end{equation}
{where $\hat{x}$ is the signal estimate.}
The SNR is computed as
\begin{equation*}
\SNR = \frac{||x||^2_2}{\operatorname{Trace}(\Sigma)},
\end{equation*}
where~$\Sigma$ is the covariance matrix of the noise term in~\eqref{Eq-1-1}. 

For each SNR value, we conducted~$100$ trials. In each trial, we sampled a signal of length $L=15$, drawn from a normal distribution $\mathcal{N}(0,I)$, and then normalized it {such that $||x||_2=1$.} 
{We use low-dimensional signals since, in this regime, the MoM has clear computational advantages over methods that maximize the likelihood function, see Section~\ref{Sc-2-1}. In addition, it was recently shown that  for high-dimensional signals, in contrast to the low-dimensional case, the sample complexity of the MRA model is not governed by moments~\cite{romanov2020multi}}.
The distribution~$\rho$ was drawn {uniformly} as an element over the simplex. {Uniform distribution is merely a point on the continuous simplex, and therefore the distribution is almost surely non-uniform.} Using the sampled signal and distribution, we generated $N=100,000$ observations according to~\eqref{Eq-1-1}.

In the experiments, we compare  the GMM estimator~\eqref{Eq-2-1} with the classical MoM, corresponding to GMM with~$W=I$; we refer to the latter as the LS estimator. We implemented both estimators using the interior point solver of MATLAB. Since the {scale of the relative error} changes drastically for different SNR levels, {we measure the ratio between the relative errors of the GMM and the LS estimators.}
In the figures {presented in this section}, {the blue and red lines represent, respectively,  the mean and median of {this} relative error ratio}. For clarification, when the ratio is greater than~1, the {LS's} relative error is bigger than the relative error of the GMM estimator. Namely, the GMM outperforms the LS. The upper and lower limits of the boxes denote the {75th}  and {25th} percentiles, respectively. The dots represent trials whose ratios lie below the first quartile or above the third quartile. Lastly, the dashed line represents {the LS's mean relative error}, which corresponds to the right vertical axis. We focus on the range of $\SNR$ {levels corresponding} to relative errors that are smaller than~$1$. The code to reproduce {this section} experiments is publicly available at {\url{https://github.com/abasasa/gmm-mra}.}
 
\subsection{Homoscedastic noise}
We {start with} a homoscedastic noise model {where},  $\varepsilon_i \overset{i.i.d.}{\sim} \mathcal{N}(0,\sigma^2I)$; this noise model was considered by all previous works on MRA, e.g.,~\cite{Abbe2017,Bendory2018,bandeira2020non}.	

\begin{figure}
	\centering
	\includegraphics[width= 0.45\textwidth]{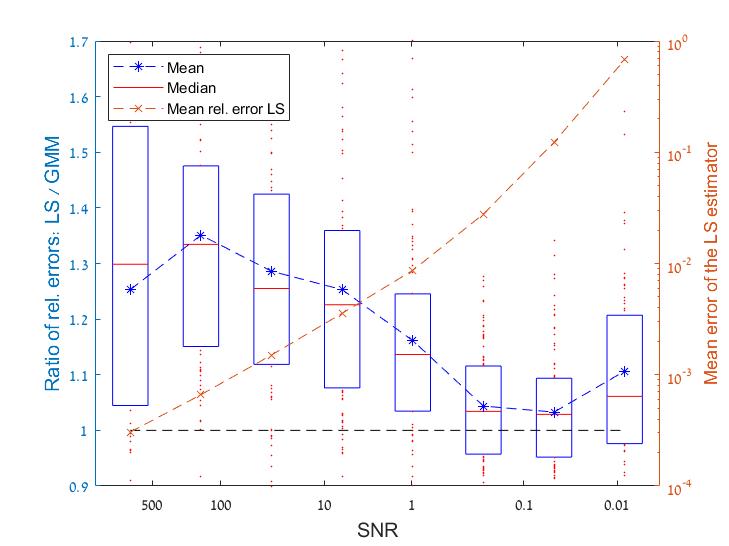}
	\caption{\textbf{Homoscedastic noise model}: The ratio between the relative errors of the LS and the GMM {estimators}.}
	\label{Fig-3-1}
\end{figure}

The first experiment compares the performance of the  LS and GMM estimators. The results are presented in Figure~\ref{Fig-3-1}. {As expected, the mean relative error of the estimators decreases as the SNR increases, as illustrated by the red dashed line.}
For high SNR levels, {the} GMM estimator outperforms the LS estimator by more than~$20\%$. In low SNR levels, the GMM estimator {has} only a slight advantage over {the} LS {estimator}. 

{
\subsection{The effect of adding the third moment}
This paper is mainly focused on using the first two moments: the minimal number of moments required  for signal estimation when the distribution of translations is non-uniform~\cite{Abbe2017}. This is also the number of moments we use for the cryo-EM experiments in Section~\ref{Sc-4}. 

We now examine a natural question: what is the effect of  adding the third moment to the moment function~\eqref{Eq-3-1}? 
Figure~\ref{Fig-M3-1} compares the performance of the GMM estimator with two and three moments. Adding the third moment indeed leads to significant improvement for all  SNR levels, especially when the SNR is high. The superior numerical performance of adding the third moment to the moment function comes at a cost: the dimensionality (and thus the computational load) grows from $O(L^2)$ to  $O(L^3)$. In particular, in our experiments, adding the third moment has increased the  running time  by a factor ranging from~1.5 to~5 (depending on the SNR). Figure~\ref{Fig-M3-2} compares the performance of the GMM and LS estimators with the third moment. The trend of the ratio of relative errors is similar (besides a single point) to the same experiment with only the first two moments, as presented in Figure~\ref{Fig-3-1}. 

 For the rest of the paper, we continue investigating the  GMM estimator using only the first two moments.
}



 
 

{

\begin{figure}[ht]
	\centering
	\includegraphics[width=.45	\textwidth]{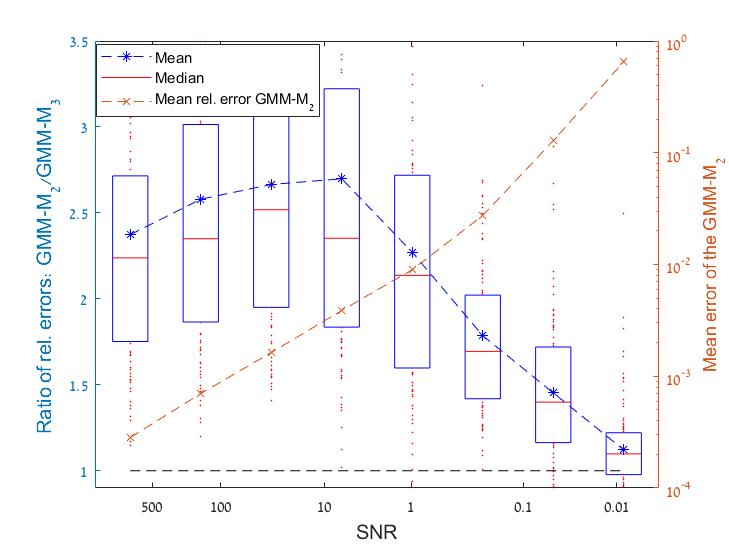}
	\caption{\textbf{The effect of adding the third moment}: The ratio between the relative errors of the GMM estimator with the first three moments (GMM-$M_3$) and the GMM estimator with only the first two moments (GMM-$M_2$). Evidently, the GMM with three moments outperforms the GMM with two moments.
	}
	\label{Fig-M3-1}
\end{figure}

\begin{figure}[ht]
	\centering
	\includegraphics[width=.45	\textwidth]{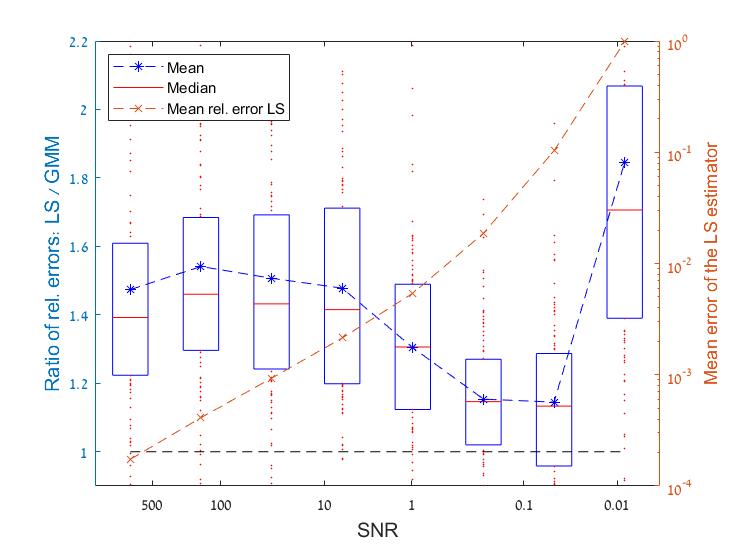}
	\caption{\textbf{Homoscedastic noise model with the third moment}: The ratio of the relative errors of the LS and the GMM estimators when the moment function includes the third moment.}
	\label{Fig-M3-2}
\end{figure}

}

\subsection{Heteroscedastic noise}
Next, we investigate a scenario in which the noise term is distributed as  $\varepsilon_i \overset{i.i.d.}{\sim} \mathcal{N}(0, \Sigma_1)$, where $\Sigma_1$ is a diagonal matrix given by
\begin{equation} \label{Eq-3-3}
\Sigma_1 = 
\begin{bmatrix}
\sigma^2 & & \\
& 2  \sigma^2 & & \\
& & \ddots & \\
& & & L \sigma^2
\end{bmatrix}.
\end{equation}
In this case, the noise level
increases along the {signal's entries}. This model is similar to a popular noise model in cryo-EM~\cite{Bendory2020}.
{Figure~\ref{Fig-3-2} compares the performance of the GMM and the LS methods and shows that the GMM outperforms the LS by at least~$30\%$ for most SNR levels. Note that the performance of the LS estimator is similar in both the homoscedastic and heteroscedastic {noise} models.}

\begin{figure}[ht]
	\centering
	\includegraphics[width=.45	\textwidth]{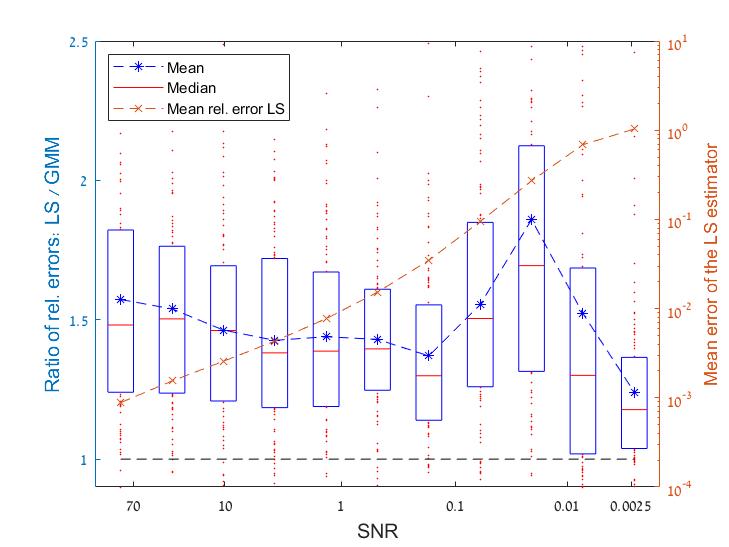}
	
	\caption{\textbf{Heteroscedastic noise model}: The ratio between the relative errors of the LS and the GMM {estimators}.}
	\label{Fig-3-2}
\end{figure}

\subsection{When do we expect the GMM to outperform the MoM?}\label{Sc-3-1}

In the heteroscedastic noise model, the dispersion of eigenvalues is {larger} than in the homoscedastic case. Therefore, the effect of including the weighting matrix~$W_N$ in~\eqref{Eq-2-1} leads to an {increased} performance gap between the LS and the GMM. Accordingly, we conjecture that in general, if~$W$ is ``far''  from the identity matrix {in terms of} {large} dispersion of eigenvalues, then the performance of the GMM will be significantly better compared to the LS estimator. 
{To this end, we define the distance between a matrix~$A\in\R^{n\times n}$ and the identity matrix {by}
	\begin{equation}\label{Eq-3-11}
		\delta(A) := \left[\sum_{j=1}^n \log^2\left(\frac{\sqrt{n}\lambda_j(A) }{||A||_\text{F}}\right)\right]^{\frac{1}{2}},
	\end{equation}
	where $\lambda_j(A)$ is the~$j$-th eigenvalue of~$A$. 
This measure is based on  the geodesic distance with respect to the Riemannian metric over the cone of PD matrices~\cite{moakher2005differential}, with one additional normalization factor $\frac{\sqrt{n}}{||A||_\text{F}}$. This factor assigns the same distance for matrices {that} are equal up to a scalar multiplication, i.e., $\delta(A) = \delta(cA)$ for any positive $c > 0$. 
}
In weighted LS, multiplication of the objective function by a scalar does not affect the estimator, {and thus} the modification is needed.

We compare the measure~\eqref{Eq-3-11} of the weighting matrix~\eqref{Eq-2-7}, for the homoscedastic and heteroscedastic noise models; the results are presented in Figure~\ref{Fig-3-5}. For the homoscedastic noise model, the minimum point of the geodesic distance is around $\text{SNR}\sim 0.1$. {Indeed, }Figure~\ref{Fig-3-1} shows that at this noise level, the performance of the GMM estimator performs similarly to the LS estimator. In addition, we observe that {in high SNR homoscedastic noise and under heteroscedastic noise}, large geodesic distance is positively correlated with superior performances of the GMM compared to LS. These observations support our hypothesis that the  {geodesic} distance from the identity can be used as a heuristic tool to predict when it would be beneficial to use the GMM.

\begin{figure}[ht]
	\centering
	\includegraphics[width=.45	\textwidth]{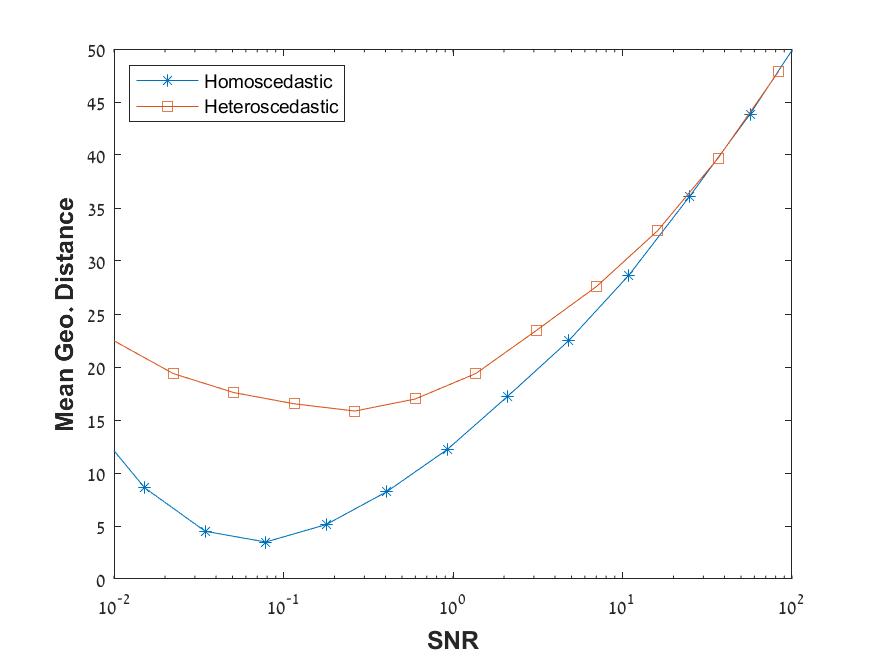}
	\caption{{Modified} geodesic distance~\eqref{Eq-3-11}  of the weighting matrix~$W$ from the identity matrix {as a function of the SNR}. The two curves are of the homoscedastic and heteroscedastic noise models, averaged over $20$ trials. 
	{	
		In the homoscedastic noise model with $\text{SNR} \approx 0.1$, where the geodesic distance between~$W$ and the identity matrix is small, the   performance of the GMM and the LS estimators is similar.}}
	\label{Fig-3-5}
\end{figure}

\subsection{MRA with outliers}
Motivated by the abundance of outliers in cryo-EM datasets~\cite{Bendory2020}, we consider the following generative model
\begin{equation}\label{Eq-3-4}
y_i = \begin{cases}
R_{S_i}x + \epsilon_i \quad \text{  w.p.  } 1-p_{\text{out}},\\
\mathcal{N}(0,\Sigma_{\text{out}}) \quad \text{w.p. } p_{out},
\end{cases}
\end{equation}
where $p_{\text{out}}$ is the probability of an observation to be an outlier and $\Sigma_{\text{out}}$ is the covariance matrix of the outliers' distribution.

Under this statistical model, the analytical moments read 
\begin{equation*}
\begin{gathered}
M_1(x,\rho) = (1-p_{\text{out}})  C_\rho x,\\
M_2(x,\rho) =(1-p_{\text{out}})  C_x D_\rho C_x^T + (1-p_{out}) \Sigma + p_{\text{out}} \Sigma_{\text{out}}.
\end{gathered}
\end{equation*}
We assume that $p_{\text{out}}$ and $\Sigma_{\text{out}}$ are known. As before, for large $N$, the empirical moments approximate the analytical moments.

\subsubsection{Weighted LAD optimization}
Robustness to outliers is commonly obtained by replacing the LS objective with the weighted \textit{least absolute deviations} (LAD) objective function~\cite{koenker1978regression}, given by
\begin{equation}\label{Eq-3-8}
\hat{\theta}_N^{LAD} =  \arg\min_{\theta \in \Theta} |{W_{\text{vec}}}| \cdot |h(\theta, \{y_i\}_{i=1}^N)|.
\end{equation}
Here, {$|W_{\text{vec}}|$} is a fixed weighting vector, $\cdot$ is the scalar product, the absolute value is taken entrywise, and $h(\theta, \{y_i\}_{i=1}^N)$ is the loss function.    
Intuitively, one might suggest {defining $h$ of~\eqref{Eq-3-8} as} $g_N$ from~\eqref{Eq-2-5}. However, it was shown in~\cite{Jong2002} that the estimation error of standard GMM based on~$\ell_2$ norm is lower than any other~$\ell_p$ norm for a moment function that satisfies $\E\left[f(\theta_0, y_i)\right] = 0$\footnote{This is not necessarily true for a biased moment function.}. 
Therefore, we take a different approach. 

\subsubsection{Geometric median estimator}
We present an estimator based {on} the geometric median of the moments. The geometric median is a generalization of the univariate median to the multidimensional case, and is defined as
\begin{equation}\label{Eq-3-6}
\widehat{GM}(\{y_i\}_{i =1 }^N) = \arg\min_{z \in \R^{q}} \sum_{i=1}^N {\left\| z - \begin{pmatrix}
y_i\\ 
y_i y_i^T
\end{pmatrix}\right\|_2},
\end{equation}
where $q = L + (L+1)L/2$, the total number of elements of the first and second moment. Notice that as before, we treat $y_i y_i^T$ as a vector in $\R^{(L+1)L/2}$. The  geometric median does not have a closed-form solution, and we approximate it using the {Weiszfeld} algorithm~\cite{weiszfeld2009point}.

We estimate the empirical moments by~\eqref{Eq-3-6}, and use it in the loss function of the weighted LAD~\eqref{Eq-3-8},
\begin{equation}\label{Eq-3-5}
h_{GM}(x, \rho, \{y_i\}_{i=1}^N) = \begin{pmatrix}
M_1(x,\rho)\\
M_2(x,\rho)
\end{pmatrix} - \widehat{GM}(\{y_i\}_{i =1 }^N).
\end{equation}
We name this estimator the \textit{geometric median estimator}.
In order to construct~{$W_{\text{vec}}$}, we use an alternating optimization scheme with two phases. The first phase minimizes~\eqref{Eq-3-5} for a fixed~{$W_{\text{vec}}$} (initialized by vectors of ones). Then, we update~{$W_{\text{vec}}$} using the current estimation of~$x$ and~$\rho$; see Algorithm~\ref{Alg-4-1}. As mentioned in Section~\ref{Sc-3-1}, a multiplication of~{$W_{\text{vec}}$} by a scalar does not affect the GMM estimator, {and} therefore we minimize~{$W_{\text{vec}}$} over the sphere in $\R^{L+(L+1)L/2}$, {denoted here by} $S^{L+(L+1)L/2-1}$, using Manopt's BFGS solver~\cite{manopt}. {Empirically,   updating~{$W_{\text{vec}}$} once suffices}.
\begin{algorithm} 
	\SetAlgoLined
	\KwOut{{The estimated parameter vector $\hat{\theta}$ }}
	\textbf{Input:}:  A set of observations $\{y_i\}_{i=1}^N$, a loss function $h$, and number of steps $n$ 
	\begin{enumerate}
		\item 	Initialize ${W_{\text{vec}}}$ (a common choice is a vector of ones)
		\item 	Solve $\hat{\theta} = \underset{\theta \in \Theta}{\operatorname{argmin}} |{W_{\text{vec}}}| \cdot |h(\theta,\{y_i\}_{i=1}^N)|$ (first iteration)
		\item \For{j from 2 to n}{			
			${W_{\text{vec}}} = 			\underset{W \in S^{L+(L+1)L/2-1}}{\operatorname{argmin}} |W|\cdot  |h(\hat{\theta},\{y_i\}_{i=1}^N)|$ \newline
			$\hat{\theta} = \underset{\theta \in \Theta}{\operatorname{argmin}}|{W_{\text{vec}}}| \cdot |h(\theta,\{y_i\}_{i=1}^N)|$
		}
	\end{enumerate}	
	\caption{{Alternating weighted LAD}}
	\label{Alg-4-1}
\end{algorithm}

\subsubsection{Numerical Results}
The geometric median estimator is compared against the GMM  in Figure~\ref{Fig-3-4} with $p_{\text{out}} = 0.2$ and $\Sigma_{\text{out}} = \frac{100}{L \cdot \text{SNR}} I$, where $L=15$ is the signal length. As can be seen, for low SNR levels, the geometric median estimator exhibits superior numerical performance, while GMM works better in a high SNR environment. We observed similar results for different~$p_{\text{out}}$ values, however the SNR level of the transition point varies.

\begin{figure}[ht]
	\centering
	\includegraphics[width= 0.45\textwidth] {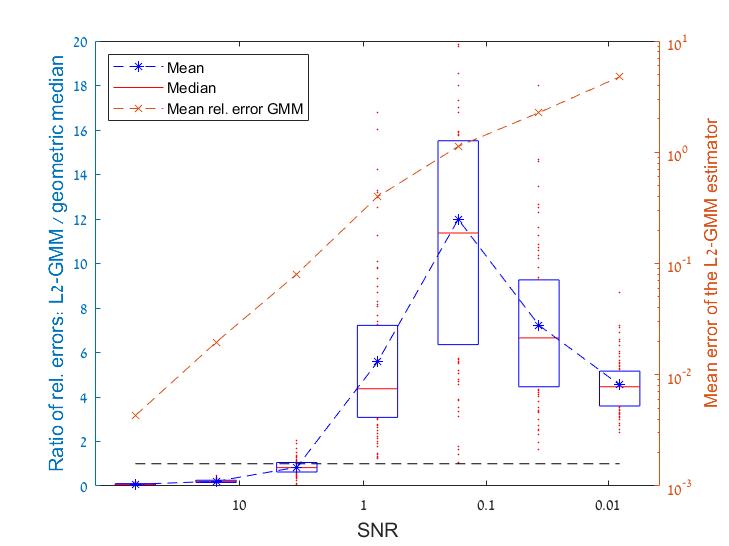}
	
	\caption{\textbf{Outliers noise model~\eqref{Eq-3-4}}: The ratio between the relative error of the GMM and the geometric median estimator for the outliers noise model~\eqref{Eq-3-4} with $p_{\text{out}} = 0.2$.} 
	\label{Fig-3-4}
\end{figure}

\subsection{MRA with projection}\label{Sc-3-2}
The MRA with projection model is an extension of the standard MRA model~\eqref{Eq-1-1}, which includes an additional {linear} {operator} acting on the shifted signal. In this model, the~$i$-th observation is given by
\begin{equation} \label{Eq-3-7}
y_i = P R_{s_i}x + \varepsilon_i,
\end{equation}
where $P$ is a fixed, known matrix of size $K \times L$. As in~\eqref{Eq-1-1}, the goal is to estimate $x\in\R^L$ from $y_1,\ldots,y_N\in \R^K$. The first and second moments of~\eqref{Eq-3-7} are given by
\begin{equation*}
\begin{split}
M_1(x,\rho) {}&= P C_\rho x,\\
M_2(x,\rho) &= P C_x D_\rho C_x^T P^T + P \Sigma P^T.
\end{split}
\end{equation*}
	Motivated by cryo-EM, we focus on a matrix~$P$ that samples only the first~$K$ entries of the shifted signal. That is, {$PR_sx$} is a vector of length~$K$ consists of the entries {$[(R_sx)_0,\ldots,(R_sx)_{K-1}]$.}

{Figure~\ref{Fig-3-3} compares the performance of the GMM and {the} LS estimators as a function of~$K$ (namely, how many entries are being kept) for a fixed noise level $\Sigma = {10^{-2} \cdot \, I_{K\times K}}$. As can be seen, the GMM outperforms the LS by at least~$20\%$.
}
%
\begin{figure}[ht]
	\centering
	\includegraphics[width=.45	\textwidth]{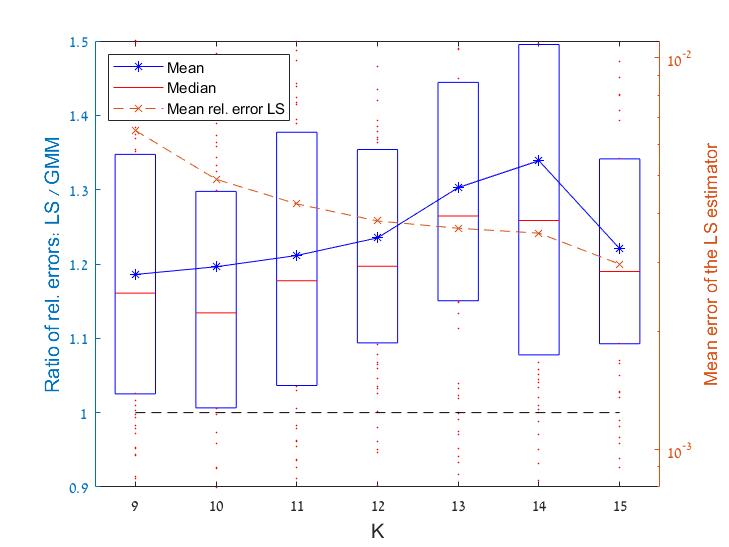}
	
	\caption{\textbf{MRA with projection:} The ratio between the relative errors of the GMM and the LS estimators as a function of the measurement length $K$. The {covariance of the noise is} $\Sigma = {10^{-2} \cdot \, I_{K\times K}}$.}
	\label{Fig-3-3}
\end{figure}

\section{GMM for cryo-EM: a proof of concept}\label{Sc-4}

As aforementioned, the main motivation of this work stems from the task of reconstructing the 3-D structure of molecular structures using single-particle cryo-EM. Building upon~\cite{Sharon2020}, we formulate the GMM for cryo-EM and show preliminary numerical results.
{We note that the goal of the MoM in the cryo-EM literature is not to reconstruct a high-resolution molecular structure, but only to quickly constitute a low-resolution, ab-initio model~\cite{Bendory2020}.}

\subsection{Mathematical background}
Under some simplifying assumptions, the cryo-EM problem involves recovering {a} 3-D  volume~$\phi\colon\R^3 \to \R$ from a set of 2-D {tomographic} projection images. Each observation (projection image) is modeled as~\cite{Frank2006}
\begin{equation}\label{Eq-4-6}
I_j = P(R_j^T\circ \phi) + \varepsilon_j,
\end{equation}
where the projection operator $P\colon\R^3 \to\R^2$ is 
\begin{equation*}
P\phi(x_1,x_2) = \int_{-\infty}^{\infty} \phi(x_1,x_2,x_3)dx_3.
\end{equation*}
{The term $\varepsilon_j$ {models} additive noise term,} and~$R_j$ is an element of the group of 3-D rotations~$SO(3)$, which can be represented as a~$3\times 3$ orthogonal matrix acting by
\begin{equation*}
[R_j^T \circ \phi(x_1,x_2,x_3)] = \phi(R_j [x_1 \ x_2 \ x_3]^T).
\end{equation*}
As in the MRA model~\eqref{Eq-1-1}, the group elements $R_j$ are unknown. We assume that each image is sampled on an~$n \times n$ Cartesian grid within the box $[-1,1]\times [-1,1]$. 

The Fourier slice theorem states that the 3-D Fourier transform of a tomographic projection is equal to a slice of the Fourier transform of the volume. Thus, the Fourier transform of~\eqref{Eq-4-6} reads 
\begin{equation} \label{Eq-4-1}
\begin{gathered}
\widehat{I}_j(z_1,z_2) = 
(R_j^T \circ \widehat{\phi})(z_1, z_2, z_3)|_{z_3 =0} + \widehat{\varepsilon}_j.
\end{gathered}
\end{equation}
The common generative model of cryo-EM includes additional complications---such as the microscope's point spread function and heterogeneous mixture of molecules~\cite{Bendory2020}---which are {disregarded} here for simplicity.

Following~\cite{Sharon2020}, we formulate the problem of recovering the 3-D structure~$\phi$ from the first two moments.
{In polar coordinates, we can write the first empirical moment as}
\[
\hat{M}_1(r_1,\varphi_1) = \frac{1}{N} \sum_{j=1}^N \widehat{I}_j(r_1,\varphi_1), \]
{and the second empirical moment}
\begin{multline*} 
\hat{M}_2(r_1,\varphi_1,r_2,\varphi_2) = \\
\frac{1}{N} \sum_{j=1}^N \widehat{I}_j(r_1,\varphi_1,r_2,\varphi_2) \widehat{I}_j(r_1,\varphi_1,r_2,\varphi_2).
\end{multline*}
The first and second moment can be computed analytically
\begin{align} \label{Eq-4-3}
\begin{split}
M_1(\phi, \rho) &= \E\left[\widehat{P(R^T \circ \phi)}\right], \\ 
M_2(\phi, \rho) &= \E\left[\widehat{P(R^T \circ \phi)} \otimes \widehat{P(R^T \circ \phi)}\right],
\end{split}
\end{align}
where the expectation is taken over {the rotations and the noise.} {The explicit expressions of $M_1$ and $M_2$ of~\eqref{Eq-4-3} {are} provided in Appendix~\ref{Ap-6}.}
\begin{figure*}[t]
	\centering
	\begin{subfigure}[b]{0.45\textwidth}
		\centering
		\includegraphics[width=\textwidth]{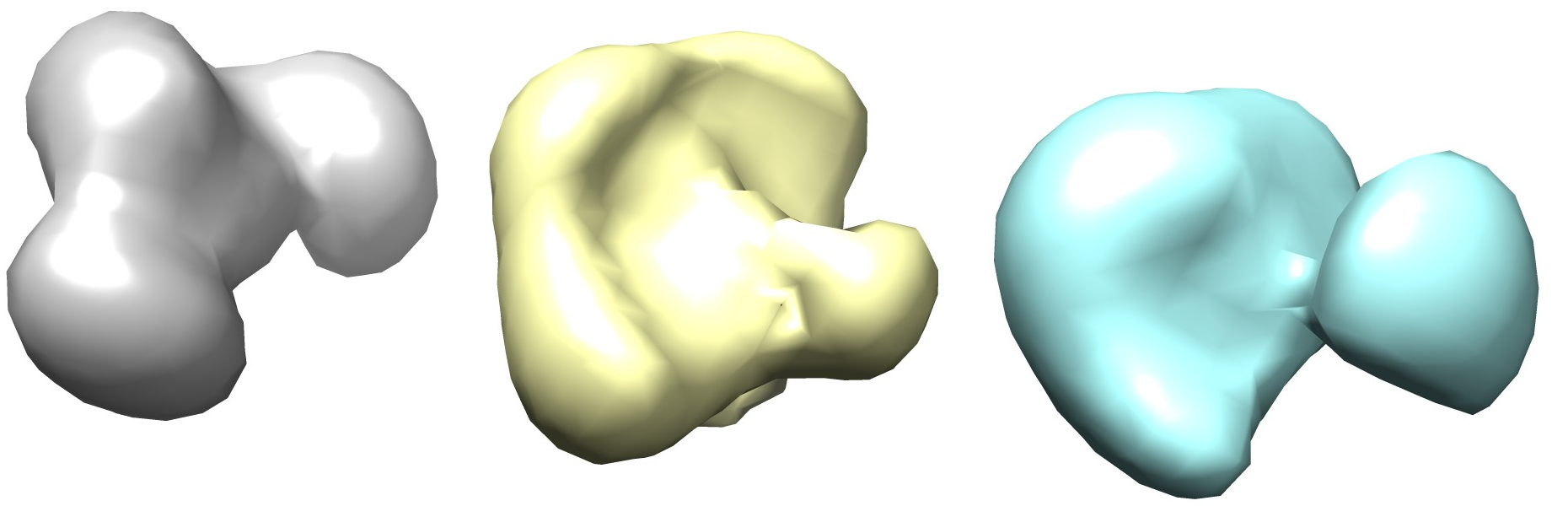}		
	\end{subfigure}
	\hfill
	\begin{subfigure}[b]{0.45\textwidth}
		\centering
		\includegraphics[width=\textwidth]{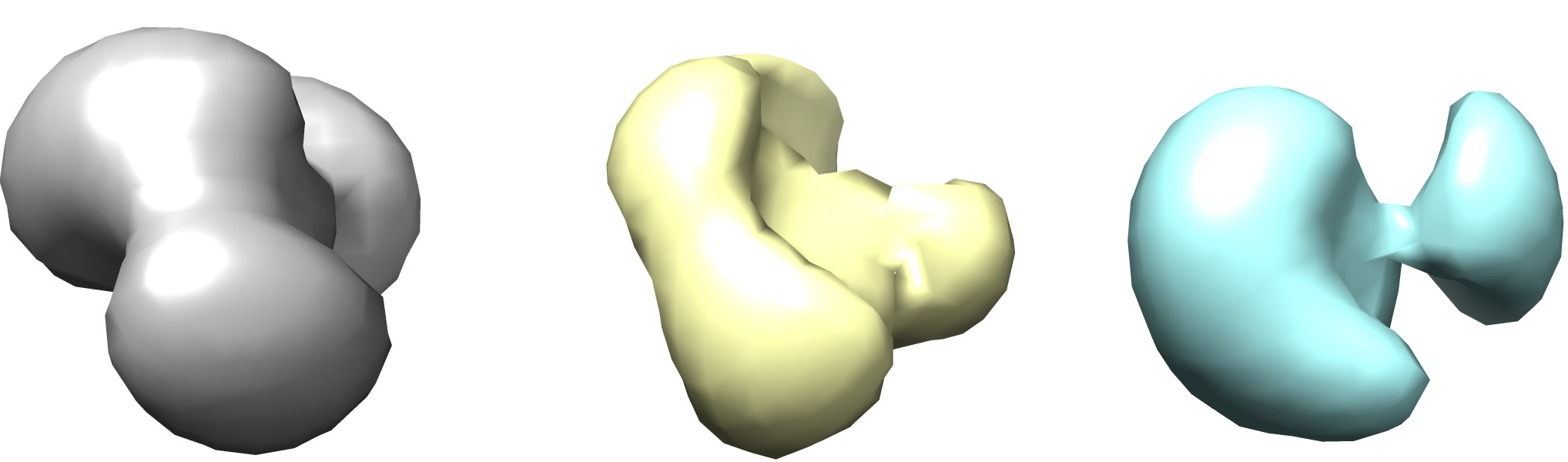}		
	\end{subfigure}
	\begin{subfigure}[b]{0.45\textwidth}
		\centering
		\includegraphics[width=\textwidth]{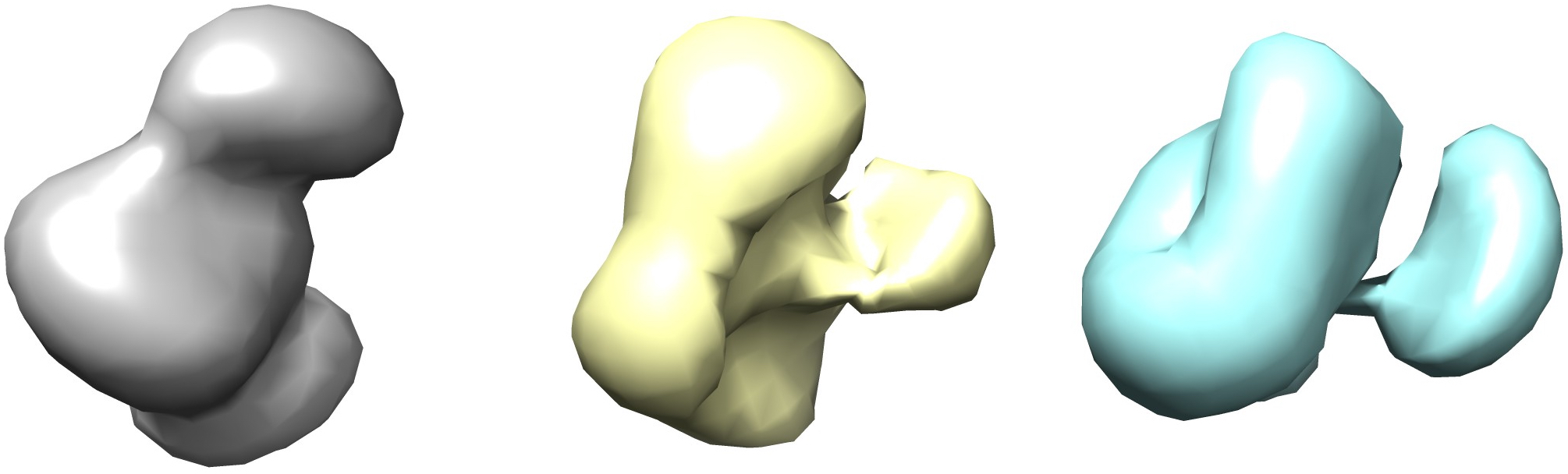}
	\end{subfigure}
	\caption{Reconstructing a 3-D structure from noisy cryo-EM projections with $\SNR {\approx} 0.4$. We present the structures from three different  viewing directions. The ground truth, the LS estimation (MoM), and the GMM estimation appear in gray, yellow, and blue, respectively. The GMM outperforms the LS in both the FSC and relative error criteria. Its FSC value and relative error are~$4.2$ and~$0.58$, respectively. The FSC value of the LS estimator is~$4.5$, and its relative error is~$0.61$. 
	}		
	\label{Fig-4-3}
\end{figure*}

\subsection{Numerical results}
\subsubsection{Evaluation metrics}
{We used two metrics: relative error and Fourier shell correlation (FSC)}. The relative error is computed  similarly to~\eqref{Eq-3-2}, except that the alignment is computed over the group of {3-D} rotations~$SO(3)$. FSC is a common resolution measure in the cryo-EM field~\cite{harauz1986exact}. It measures cross-correlation coefficients between two 3-D volumes over corresponding shells in Fourier domain. Given two volumes, $\varphi_1$ and $\varphi_2$, the FSC in a shell~$\kappa$ is calculated using all voxels~$\boldsymbol{\kappa}$ on this~$\kappa$-th shell:
\begin{equation*}
\operatorname{FSC}(\kappa) = \frac{\sum_{||\boldsymbol{\kappa}|| = \kappa}|\varphi_1(\boldsymbol{\kappa})||\overline{\varphi_2(\boldsymbol{\kappa})}|}{\sqrt{\sum_{||\boldsymbol{\kappa}|| = \kappa} |\varphi_1(\boldsymbol{\kappa})|^2  \sum_{||\boldsymbol{\kappa}|| = \kappa} |\varphi_2(\boldsymbol{\kappa})|^2}} .
\end{equation*}
Usually, FSC curves decrease with~$\kappa$, and the resolution is determined as the point where the FSC curves drop below a pre-specified value.
In this work, we use a threshold of~$0.5$~\cite{VANHEEL2005250}. Therefore, higher resolution is indicated by a smaller FSC value.

\subsubsection{Example}
Following~\cite{Sharon2020}, the LS estimator is formulated as 
\[\arg\min_{\phi, \rho} \|M_1(\phi, \rho) - \hat{M}_1 \|_\text{F}^2 + \lambda \|M_2(\phi, \rho) - \hat{M}_2\|_\text{F}^2,\] where~$\lambda$ is a regularization parameter. As in the MRA problem, we define the moment function of the GMM estimator as
\begin{equation} 
f(x,\rho, I_j) = 
\begin{pmatrix}
M_1(\phi, \rho) - \hat{I}_j\\
M_2(\phi, \rho) - \hat{I}_j^{\otimes 2} + B
\end{pmatrix},
\end{equation}
where $B$ is an unbiasing term. We removed redundant entries, due to the inherent symmetries of the moments. In order to estimate the volumes, we used Manopt's trust-regions solver~\cite{manopt}.

In the following experiment, the volume is a toy model, composed {of} five Gaussians over~$\R^3${, whose high-frequencies were removed (see the blue volumes in Figure~\ref{Fig-4-3})~\cite{Sharon2020}}. The toy model is sampled on a~$23\times 23 \times 23$ Cartesian grid. The distribution of rotations~$\rho$ was drawn randomly such that~$\rho$ is invariant to in-plane rotations, i.e., $\rho$ depends only on the viewing direction. We generated $N = 200,000$ observations according to~\eqref{Eq-4-6},
	{and added an i.i.d. white Gaussian noise, corresponding to $\SNR \sim 0.4$.}
	 Due to the simulated nature of the volume, the FSC's resolution {units} are measured in $1/\text{pixel's length}$.

The resulted volumes are depicted in Figure~\ref{Fig-4-3}. {The} GMM outperforms the LS in both criteria. Its FSC value and relative error are~$4.2$ and~$0.58$, respectively. The FSC value of the LS estimator is~$4.5$, and its relative error is~$0.61$. {The code is publicly available at~\url{https://github.com/ABASASA/GMM-Cryo}.}

\subsection{Ill-conditioning of {the} weighting matrix}\label{Sc-4-4}
The condition number of the GMM's ideal weighting matrix~$W$ in the MRA model was around~$10^4$. Unfortunately, the condition number of~$W$ in the cryo-EM experiments was much higher. Figure~\ref{Fig-4-1} presents both the condition number of~$W$ and the geodesic distance between~$W$ and the identity matrix. As {can} {be} seen, as the {SNR} increases, so are the condition number and the distance to the identity matrix. In particular, for most SNR levels, $W$ is  ill-conditioned.
{This phenomenon must be considered when applying GMM to cryo-EM experimental datasets, as discussed in the next section.}

\begin{figure}
	\centering
	\includegraphics[width= 0.45\textwidth]{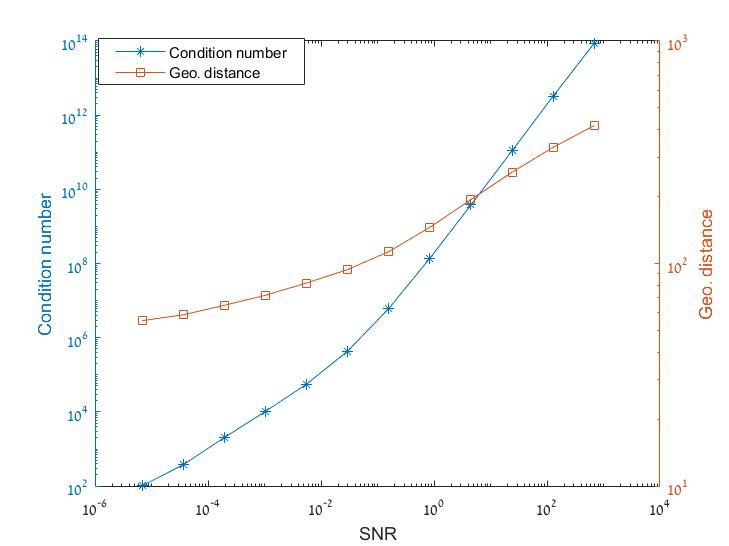}
	\caption{The condition number and the geodesic distance~\eqref{Eq-3-11}, for the cryo-EM model~\eqref{Eq-4-6}, of the GMM weighting matrix (computed at the ground truth) for different levels of white noise.}
	\label{Fig-4-1}
\end{figure}

\section{Discussion}

This paper is part of an ongoing effort to harness the favorable computational properties of the method of moments for constituting ab initio models of molecular structures using cryo-EM~\cite{Bendory2018a, Sharon2020,kam1980reconstruction, levin20183d,lan2021random}. We use several MRA models as test cases and show that the GMM outperforms the classical method of moments. Also, we prove that the GMM retains its optimal statistical properties even when the signal is determined up to a group action, as in MRA and  cryo-EM.

Our ultimate goal is to apply the GMM technique to constitute 3-D molecular structures using cryo-EM. 
 However, our study reveals a major computational challenge: the moment function’s covariance matrix is severely ill-conditioned. Thus, devising rigorous techniques to reduce the condition number is a future research direction toward applying the GMM framework to experimental cryo-EM datasets.

In a broader perspective, we intend to study the GMM for additional signal processing tasks in which the method of moments plays an important role. Examples include super-resolution~\cite{bendory2020super} and multi-target detection~\cite{bendory2019multi, lan2020multi,bendory2021multi}.

\section*{Acknowledgment}
{We are grateful to  Prof. Arie Yeredor for his deep insights that significantly improved this work.} 

\bibliographystyle{unsrt}

\appendix

\subsection{Additional assumptions of the GMM}\label{Ap-1}
{Note that by construction, $g_N$ of~\eqref{Eq-2-5} satisfies $\lim_{N\to\infty} g_N(\theta) =\E\left[f(\theta, y)\right]:= g(\theta).$}
In addition, we denote {the corresponding quadratic forms,}
\begin{equation*}
\begin{gathered}
Q(\theta) = g(\theta)^T  W  g(\theta), \quad
Q_N(\theta) = g_N(\theta)^T  W_N  g_N(\theta), 
\end{gathered}
\end{equation*}
{ and the Jacobians,
\begin{equation*}
\begin{gathered}
G(\theta) = \E\left[\partial f(\theta, y) / \partial \theta^T\right], \
G_N(\theta) =\frac{1}{N} \sum_{i=1}^N \partial f(\theta, y_i) / \partial \theta^T.
\end{gathered}
\end{equation*} }
Notice that {we also have} $\lim_{N\to\infty} Q_N(\theta) = Q(\theta)$ and $\lim_{N\to\infty} G_N(\theta) = G(\theta)$.

The favorable statistical properties of the GMM rely on {several} conditions, including the uniqueness of a set of parameters~\eqref{Eq-2-4}, see Section~\ref{Sc-2.5}. For completeness, we {now} present  the rest of the regularity conditions {required for Theorem~\ref{Thm-2-6}}~\cite{Hansen1982,Hall2005}. {For the first assumption, we say that a sequence of random variables with a joint cumulative distribution function $F$ is strictly stationary if, for any $\tau$, $n$, $t_1, \ldots, t_n \in \mathbb{N}$,
\[ F(y_{t_1}, y_{t_2},..., y_{t_n}) = F(y_{t_1 + \tau}, y_{t_2 + \tau},..., y_{t_n + \tau}). \]
The above sequence (or process) of random variables is called ergodic if 
$\E\left[ ||\frac{1}{N}\sum_{i = 1}^N y_i-\E\left[y_i\right] ||^2 \right] \to 0 $, as $N\to \infty$.}

\begin{assumption} \label{As-Ap1-1*}
	The sequence $\{y_i \mid i \in \mathbb{N}\}$ {is strictly stationary and ergodic.}	
\end{assumption}

\begin{assumption} \label{As-Ap1-2*}
	(i) $\Theta$ is a compact space and (ii) $\theta_0$ is an interior point in $\Theta$.	
\end{assumption}

\begin{assumption}[Regularity conditions for $f$]\label{As-Ap1-3*}
	The moment function, $f\colon \Theta \times  \mathbb{R}^r \to \mathbb{R}^q$, satisfies: (i) $f(\cdot, y)$ is continuous on $\Theta$ for every $y \in \mathbb{R}^r$; (ii) $g(\cdot)$ exists and finite for every $\theta \in \Theta$; (iii) $g(\cdot)$ is continuous on $\Theta$ and (iv) $\E\left[\sup_{\theta \in \Theta} \| f(\theta, y) \|\right]$ is finite.
\end{assumption}

\begin{assumption}[Properties of the weighting matrix]  \label{As-Ap1-5*}
	$\{W_N\}_{N=1}^\infty$ is a sequence of PSD matrices, and~$W$ is a PD matrix such that $W_N \overset{p}{\to} W$.
\end{assumption}

\begin{assumption}[Regularity conditions for $\partial f/\partial \theta^T$]\label{As-Ap1-4*}
	(i) The Jacobian matrix $\partial f(\theta, y)/\partial \theta^T$ exists and continuous on~$\Theta$ for every $y \in \mathbb{R}^r$ and (ii) $\E\left[\partial f(\theta, y)/\partial \theta^T\right]$ exists and finite.
\end{assumption}

\begin{assumption}[{Existence of the covariance}]\label{As-Ap1-6*} 
	(i) $\E\left[f(\theta, y)f(\theta, y)^T\right]$ exists and finite; in addition (ii) {$S$ of~\eqref{eqn:cov_mat_S} exists and it is a PD matrix}.
\end{assumption}

 \begin{assumption} [Continuity of $G$] \label{As-Ap1-7*} 
  $G(\theta)$ is continuous on some neighborhood $N_\delta$ of $\theta_0$ and its symmetries under the group $\group$, given by
	\begin{equation*}
	\begin{split}
	N_\delta(\theta_0) & = \{\theta \in \Theta: \exists a \in \group, \| \theta - a \circ \theta_0 \| < \delta \}  =  \\
	& =\cup_{a \in \group} \{\theta \in \Theta:  \| \theta - a \circ \theta_0 \| < \delta \}.
	\end{split}
	\end{equation*}
\end{assumption}
\begin{assumption}[Uniform Convergence of $G_N$] \label{As-Ap1-8*}
	For $N_\delta$ of Assumption~\ref{As-Ap1-7*}, $G_N(\theta)$ satisfies
	\[\sup_{\theta \in N_\delta} \|G_N(\theta) -  \E[f(\theta, y) / \partial \theta^T]\| \overset{p}{\to} 0.\]
\end{assumption}

\subsection{Proof of Theorem~\ref{Thm-2-1}}\label{Ap-2}

In this section, we {prove} consistency for the GMM estimator under Assumptions~\ref{As-2-2} and~\ref{As-2-1}. The proof is based on similar arguments as in~\cite{Fan2018} and~\cite{Hall2005}.

\begin{lemma}[Uniform Convergence in Probability of $Q_N(\theta)$] \label{Lemma-Ap2-1}
	Under Assumptions \ref{As-2-2}, \ref{As-2-1} and Assumptions \ref{As-Ap1-1*}, \ref{As-Ap1-2*}, \ref{As-Ap1-3*} and \ref{As-Ap1-5*}, we have
{	\begin{equation} \label{eqn:Qn_convergence}
	\sup_{\theta \in \Theta} |Q_N(\theta) - Q(\theta)| \overset{p}{\to} 0.
	\end{equation} }
\end{lemma}
\begin{proof}
	Since $g(\cdot) = \E\left[f(\cdot, y)\right]$, we have
	\[\| g_N(\theta) - \E\left[f(\theta, y)\right] \|_2 \to {0}, \quad  \theta \in \Theta.\]
	By definition,
{		\begin{multline} \label{Ap-2-1}
		\sup_{\theta \in \Theta} |Q_N(\theta) - Q(\theta)| = \\
		\sup_{\theta \in \Theta} | g_N^T W_N g_N - \E\left[f(\theta, y)\right]^T W \E\left[f(\theta, y)\right] |.
		\end{multline} 
	Since the terms {of} the right hand side of~\eqref{Ap-2-1} are bounded and arbitrary small as $N \to \infty$, we obtain~\eqref{eqn:Qn_convergence}, as required. }
\end{proof}

\begin{proof} [Proof of consistency]
	Let $\delta > 0$ be a small enough constant such that Assumption~\ref{As-Ap1-7*} holds. {By definition}, $N_\delta(\theta_0)$ is in the {$\delta$-neighborhood of~$\theta_0$ and {its} orbit. T}hen, it is an open set as an union of open sets under the standard metric in~$\mathbb{R}^p$. Therefore, {its complementary set} $[N_\delta(\theta_0)]^C$ is a closed set in $\Theta$. Since $[N_\delta(\theta_0)]^C$ is a closed set in a compact set $\Theta \subset \mathbb{R}^p$, by Assumption \ref{As-Ap1-2*}, $[N_\delta(\theta_0)]^C$ {is also compact}. According to Assumption \ref{As-Ap1-3*}, the  function $g$ is continuous, {and} therefore $Q(\theta)$ is continuous as well. {Using} the extreme value theorem, $Q(\theta)$ has a minimum over a compact space $[N_\delta(\theta_0)]^C$.
	
	Let us define $\varepsilon := \inf_{\theta \notin N_\delta(\theta_0)} Q(\theta)$, the infimum of the asymptotic objective function outside the neighborhood.
	First, we prove that $\varepsilon > 0$.
	By Assumption~\ref{As-Ap1-5*}, $W$ is a constant PD matrix. Then, $Q(\theta) = 0  \iff g(\theta) =0$. However, by Assumption~\ref{As-2-2}, $g(\theta) =0$ only for a~$\theta$ in the orbit of~$\theta_0$, hence $\theta \in N_\delta(\theta_0)$. Therefore, $\varepsilon >0$.

	By Lemma~\ref{Lemma-Ap2-1},
	$\sup_{\theta \in \Theta}|Q_N(\theta) - Q(\theta)| \overset{p}{\to} 0$, so we can choose a large enough~$N$ such that: $ |Q_N(\theta) - Q(\theta)| \leq \frac{\varepsilon}{3}$. In addition, by definition $Q_N(\hat{\theta}_N) \leq Q_N(\theta_0)$, since  $\hat{\theta}_N$ is the minimum of $Q_N$. Then, for large enough $N$ we get: 
	\begin{equation*}
	Q(\hat{\theta}_N) \leq Q_N(\hat{\theta}_N) + \frac{\varepsilon}{3} \leq
	Q_N(\theta_0) + \frac{\varepsilon}{3} \leq Q(\theta_0) + \frac{2 \varepsilon}{3} < \varepsilon.
	\end{equation*}
	In other words, $\hat{\theta}_N \in N_\delta(\theta_0)$. Since $\delta$ was chosen arbitrarily, we conclude that there exists a sequence $\{a_N \}_{N=1}^\infty \subset \group$ such that $a_N \circ \hat{\theta}_N \overset{p}{\to} \theta_0$.
\end{proof}

\subsection{Proof of Corollary~\ref{Thm-2-4}}\label{Ap-3}
The following proofs extend the proofs of~\cite{Hall2005} to uniqueness of an orbit {solution}. It is done by replacing $\hat{\theta}_N$ by $a_N \circ \hat{\theta}_N$.

\begin{proof} [Proof of asymptotic normality of the parameters estimator]
	By Assumption~\ref{As-2-1}, we have $G_N(a_N \circ \hat{\theta}_N) = G_N(\hat{\theta}_N)$. From Theorem~\ref{Thm-2-1} and the mean value theorem, there exists $a_N \in \group$ such that:
	\begin{equation} \label{Eq-Ap2-1}
		\begin{split}
			g_N(a_N \circ \hat{\theta}_N) {}& =  g_N(\theta_0) \\
		& + G_N(a_N \circ \hat{\theta}_N^\lambda) (a_N \circ \hat{\theta}_N - \theta_0).
		\end{split}	
	\end{equation}
	Here, $\hat{\theta}_N^\lambda$ exists according to the mean value theorem. Note that the i-th row of $G_N(a_N \circ \hat{\theta}_N^\lambda)$ is the corresponding row to $G_N({\beta}^{(i)})$, where ${\beta}^{(i)} = \lambda_{N,i} \theta_0^{(i)} + (1 - \lambda_{N,i})[a_N \circ \hat{\theta}_N]^{(i)}$. Multiplying~\eqref{Eq-Ap2-1} by $G_N(\hat{\theta}_N)^T  W_N = G_N(a_N \circ \hat{\theta}_N)^T  W_N$ yields
	\begin{multline*}
	G_N(\hat{\theta}_N)^T  W_N g_N(\hat{\theta}_N) = 
	G_N(a_N \circ \hat{\theta}_N)^T W_N  g_N(\theta_0)  \\ 
	+ G_N(a_N \circ \hat{\theta}_N)^T W_N  G_N(a_N \circ \hat{\theta}_N^\lambda) (a_N \circ \hat{\theta}_N - \theta_0).
	\end{multline*}
	By definition, $\hat{\theta}_N$ is the minimum of~$Q_N(\cdot)$, and therefore $G_N(\hat{\theta}_N)^T W_N g_N(\hat{\theta}_N) = 0$. Thus,
	\begin{equation}\label{Eq-Ap3-3}
	\begin{split}
		{}&\sqrt{N} (a_N \circ \hat{\theta}_N - \theta_0) = 
		-[G_N(a_N \circ \hat{\theta}_N)^T W_N \\
		&G_N(a_N \circ \hat{\theta}_N^\lambda))]^{-1} G_N(a_N \circ \hat{\theta}_N)^T   W_N  \sqrt{N} g_N(\theta_0).
	\end{split}
	\end{equation}
	From the consistency of the estimator, we obtain:
	\begin{equation*}
	\begin{split}
	 {}&G_N(a_N \circ \hat{\theta}_N) \overset{p}{\to} G_0, \\
	 &G_N(a_N \circ \hat{\theta}_N^\lambda) \overset{p}{\to} G_0.
	\end{split}
	\end{equation*}
	Next, we denote
	\begin{equation*}
	\begin{split}
		M_N :=[{}&G_N(a_N \circ \hat{\theta}_N)^T W_N 
	G_N(a_N \circ \hat{\theta}_N^\lambda))]^{-1} \\
	\times & G_N(a_N \circ \hat{\theta}_N)^T   W_N  \sqrt{N},
	\end{split}
	\end{equation*}
	hence
	\begin{equation*}
		\lim_{N \to \infty} M_N=[G_0^T W  G_0]^{-1} G_0^T  W.
	\end{equation*} 
	The right hand side of~\eqref{Eq-Ap3-3} can be expressed as $-M_N g_N(\theta_0)$. Finally, from Slutsky’s theorem and {the} central limit theorem, we derive the desired results.
\end{proof}

\begin{proof} [Proof of optimal choice of the weighting matrix]
	Let $\hat{\theta}_N(W)$ be the GMM estimator with the weighting matrix~$W$. We denote by~$V(W)$  the variance of the limiting distribution of $\sqrt{N}[a_N \circ \hat{\theta}_N(W) - \theta_0]$. We prove that $V(S^{-1})$ is the minimum asymptotic variance, which is equivalent to prove that $V(W) - V(S^{-1})$ is a PSD matrix, for any PD matrix~$W$.
	
	According to Theorem~\ref{Thm-2-1}, we have
	\begin{equation*}
		\begin{split}
		\sqrt{N}[a_N \circ \hat{\theta}_N(W) - \theta_0] {}&\to V(W), \\
		\sqrt{N}[b_N \circ \hat{\theta}_N(S^{-1}) - \theta_0] &\to V(S^{-1}),
		\end{split}
	\end{equation*} 
	when $N\to \infty$. Here, $\{a_N\}_{N=1}^\infty$ and $\{b_N\}_{N=1}^\infty$ are sequences in~$\group$ such that $a_N \circ \hat{\theta}_N(W) \overset{p}{\to} \theta_0$ and $b_N \circ \hat{\theta}_N(S^{-1}) \overset{p}{\to} \theta_0$.  We start by relating $\sqrt{N}[a_N \circ \hat{\theta}_N(W) - \theta_0]$ and $\sqrt{N}[b_N \circ \hat{\theta}_N(S^{-1}) - \theta_0]$ by
	\begin{equation} \label{Eq-Ap5-1}
	\begin{split}
	\sqrt{N}[a_N \circ \hat{\theta}_N(W) - \theta_0] = \sqrt{N}[b_N \circ \hat{\theta}_N(S^{-1}) - \theta_0 &{}] \\
	 + \sqrt{N}[a_N \circ \hat{\theta}_N(W)- b_N \circ \hat{\theta}_N(S^{-1})&].
	\end{split}
	\end{equation}
	From the proof of Corollary~\ref{Thm-2-4}-A, 
	\begin{equation*}
	\sqrt{N}[a_N \circ \hat{\theta}_N(W) - \theta_0] = - \sqrt{N} M(W)  g_N(\theta_0) + o_p(1),
	\end{equation*}
	where $M(W) = [G_0^T W G_0]^{-1} G_0^T W$ and~$o_p(1)$ represents a sequence of vectors $\{a_N\}_{N=1}^\infty \subset\R^{p\times 1}$ such that for each entry $\lim_{N\to \infty} (a_N)_{i} =0$. Therefore,
	\begin{equation*}
	\begin{split}
	\sqrt{N}[a_N \circ \hat{\theta}_N(W)- b_N \circ \hat{\theta}_N(S^{-1})] =  \\ -\sqrt{N}[M(W)-M(S^{-1})]
	g_N(\theta_0) + o_p(1).
	\end{split}
	\end{equation*}
	Now, we apply the asymptotic covariance of~\eqref{Eq-Ap5-1} and get:
	\begin{equation*}
	V(W) - V(S^{-1}) = V_1 + C + C^T,
	\end{equation*}
	where $V_1 = \lim_{N \to \infty} \operatorname{Var}\left[[M(W)-M(S^{-1})]\sqrt{N}g_N(\theta_0)\right],$ and
	\begin{equation*}
	\begin{split}
		C =	\lim_{N \to \infty} \Cov [
		{}&\sqrt{N}[M(W)-M(S^{-1})] 
		g_N(\theta_0),\\ &\sqrt{N}M(S^{-1})g_N(\theta_0)].
	\end{split}
	\end{equation*}
	By construction, $V_1$ is a PSD matrix. All that left is to prove that~$C$ is a PSD matrix and is not a PD matrix. By definition of the covariance,
	\begin{equation*}
	\begin{split}
	{}&C   = \\ 
	&\lim_{N \to \infty} \E\left[[M(W)-M(S^{-1})]  N  g_N(\theta_0) g_N(\theta_0)^T M(S^{-1})^T\right] \\
	&=\left[M(W) -M(S^{-1})\right] \lim_{N \to \infty} \E \left[N g_N(\theta_0) g_N(\theta_0)^T   \right] \\
	& M(S^{-1})^T = M(W) S M(S^{-1})^T - M(S^{-1}) S M(S^{-1})^T.
	\end{split}
	\end{equation*}
	{Hence, for the choice of $W = S^{-1}$,
	\begin{equation*}
		\begin{split}
			C   = M(W) S M(S^{-1})^T - M(S^{-1}) S M(S^{-1})^T = 0,
		\end{split}
	\end{equation*}
	 and we get the desired result.}
\end{proof}

\subsection{Computational framework for cryo-EM} \label{Ap-6}
We follow the framework developed in~\cite{Sharon2020}. This section {provides} the {complementary} details which {were} omitted from Section~\ref{Sc-4}.

\subsubsection{Basis for the volume $\widehat{\phi}$}
We represent the Fourier transform of $\phi$ by 
\begin{equation} \label{Eq-4-5}
\widehat{\phi}(\kappa, \theta, \varphi) = \sum_{\ell=0}^{L} \sum_{m=-\ell}^{\ell}\sum_{s=1}^{S(\ell)} A_{\ell,m,s} F_{\ell,s}(\kappa) Y^m_\ell(\theta,\varphi),
\end{equation}
where $Y^m_\ell(\theta,\varphi)$ are the complex spherical harmonics, $L$ is the volume's bandlimit, and $A_{\ell,m,s}$ are the expansion coefficients. The radial frequency functions $F_{\ell,s}$ are orthogonal for each fixed $\ell$, where $s = 1,\ldots, S(\ell)$ is referred to as the radial index. Popular choices of radial functions include the spherical Bessel functions~\cite{Andrews1997}, which are eigenfunctions of the Laplacian on a closed ball with Dirichlet boundary conditions, as well as the radial components of 3-D prolate spheroidal wave functions~\cite{slepian1964prolate}. 

We assume that the volume is band-limited with Fourier coefficients supported within the radius of size~$\frac{\pi n}{2}$, where~$n$ is the size of the Cartesian grid of the projection images~$I_j$, $n\times n$.  Note that since~$\phi$ is real valued, its Fourier transform is conjugate-symmetric, which imposes restrictions on the coefficients~$A_{\ell,m,s}$.

The advantage of expanding $\widehat{\phi}$ in terms of spherical harmonics is that the space of degree $\ell$ spherical harmonics is closed under rotation. {In particular,} rotating a spherical harmonic by $R\in SO(3)$ can be expressed as 
\begin{equation}
\begin{split}
R^T \cdot Y_\ell^m(x) {}&= Y_\ell^m(Rx) \\
&=\sum_{m'=-\ell}^{\ell}U^{\ell}_{m,m'}(R)Y_\ell^{m'}(x), \ x\in S^2,
\end{split}
\end{equation}
where $U^\ell (R) \in \mathbb{C}^{(2\ell+1)\times (2\ell+1)}$ are the Wigner matrices (see~\cite[p.~343]{chirikjian2016harmonic}).

\subsubsection{Basis for the probability distribution of  rotations}
We assume the probability density $\rho$ over  $SO(3)$ is a smooth, band-limited function, which can be expressed as
\begin{equation} \label{Eq-4-4}
\rho(R) = \sum_{p=0}^{P}\sum_{u,v=-p}^{p} B_{p,u,v} U^p_{u,v}(R), \ R\in SO(3).
\end{equation}
By the Peter-Weyl theorem, $\{U^\ell(R)\}_{\ell=0}^L$ form an orthonormal basis of $L^2(SO(3))$.  The cutoff $P$ is the band limit of $\rho$. Following the arguments from~\cite{Sharon2020}, we assume that $P \le 2L$.

\subsubsection{Basis for the 2-D images $\widehat{I}_j$}
Next, we represent $\widehat{I}_j$ using a function space which is closed {under} in-plane rotations, represented as elements of  $SO(2)$. By the Peter-Weyl theorem, we can can expand a band-limited image $\widehat{I}_j$:
\begin{equation}
\widehat{I}_j(\kappa,\varphi) = \sum_{q=-Q}^{Q}\sum_{t=1}^{T(q)}a^j_{q,t} f_{q,t}(\kappa)e^{iq\varphi}.
\end{equation}
Here, the radial frequency functions $f_{q,t}$, for fixed $q$, are taken to be an orthonormal basis. Specifically, we choose  $f_{q,t}$ to be the radial components of the 2-D prolate spheroidal wave functions~\cite{slepian1964prolate}. Following~\cite{Sharon2020}, we take $Q = L$.

\subsubsection{Representation of the moments}
{We present  the connection} between the first two moments of the observed images and the coefficients $\{A_{\ell,m,s}\}_{\ell,m,s}$ and $\{B_p,u,v\}_{p,u,v}$ of the volume and the distribution of rotations, respectively.

We index the images in terms of $R \in SO(3)$ (instead of~$j$ in~\eqref{Eq-4-6}):
\begin{equation}
\widehat{I_R}(\kappa, \varphi) = \sum_{q=-Q}^Q\sum_{t=1}^{T(q)} a^R_{q,t} f_{q,t} e^{iq\varphi}.
\end{equation}
Using the Fourier slice theorem and a few algebraic steps, it can be shown that the 2-D and the 3-D coefficients are related via
\begin{equation}\label{Eq-4-7}
a^R_{q,t}  = \sum_{\ell=|q|}^{L} \sum_{s=1}^{S(\ell)}\sum_{m=-\ell}^{\ell} A_{\ell,m,s} U^{\ell}_{m,q}(R) \gamma^{q,t}_{\ell,s},
\end{equation}
where $\gamma^{q,t}_{\ell,s}$ are constants depending on the radial functions:
\begin{equation}
\gamma^{q,t}_{\ell,s} = \frac{1}{2\pi} \int_{0}^{\infty}\int_{0}^{2\pi}Y^q_{\ell}(\frac{\pi}{2}, \varphi)e^{-iq\varphi}F_{\ell,s}(\kappa)f_{q,t}(\kappa)\kappa d\kappa d\varphi.
\end{equation}
In practice, the coefficients $\gamma^{q,t}_{\ell,s}$ are calculated via numerical integration over a closed segments~\cite{Sharon2020, Lederman2017}.

\subsubsection{First Moment}
By taking the expectation over $R \sim \rho$ and applying the distribution's expansion~\eqref{Eq-4-4}, {we get}
\begin{equation} \label{Eq-4-8}
\E\left[a^R_{q,t}\right] = \sum_{\ell=|q|}^{\min(L,P)} \sum_{s=1}^{S(\ell)}\sum_{m=-\ell}^{\ell}A_{\ell,m,s}B_{\ell,-m,-q}\gamma^{q,t}_{\ell,s}\frac{(-1)^{m+q}}{2\ell+1}.
\end{equation}
{We note that} the first moment is linear in both the {volume's coefficients $\{A_{\ell,m.s}\}_{\ell,m.s}$ and the distribution $\{B_{p,u,v}\}_{p,u,v}$, as in the MRA model.}

\subsubsection{Second Moment}
We have:
{
\begin{equation}\label{Eq-4-10}
\begin{split}
\E\left[a^R_{q_1,t_1} a^R_{q_2,t_2}\right] =
 \sum_{\begin{gathered}
	\scriptscriptstyle{\ell}_1,m_1,s_1,\\ \scriptscriptstyle{\ell}_2,m_2,s_2
	\end{gathered}}
	A_{\ell_1,m_1,s_1} A_{\ell_2,m_2,s_2} \gamma^{q_1,t_1}_{\ell_1,s_1} \gamma^{q_2,t_2}_{\ell_2,s_2} \times\\
 \sum_p B_{p,-m_1-m_2,-q_1-q_2} C^{q_1,q_2}_p(\ell_1,\ell_2,m_1,m_2) \frac{(-1)^{m_1 + m_2}}{2p+1}. 
\end{split} 
\end{equation} }
The first summation have the same range as in~\eqref{Eq-4-8} for each set $(\ell_i,s_i,m_i)$. The second summation's range is \[\max(|\ell_1-\ell_2|,|m_1+m+2|, |q_1+q+2|) \le p \le \min(\ell_1+\ell_2, P).\] In addition:
\begin{equation} \label{Eq-4-9}
\begin{gathered}
{C^{q_1,q_2}_p(\ell_1,\ell_2,m_1,m_2) }=  \\
 C(\ell_1, m_1; \ell_2,m_2| p, m_1+m_2) C(\ell_1, q_1; \ell_2,q_2| p, q_1+q_2),
\end{gathered}
\end{equation}
is the product of two Clebsch-Gordon coefficients. The second moment is a quadratic function of the volume's coefficients $\{A_{\ell,m.s}\}_{\ell,m.s}$ and is linear in the distribution's coefficients $\{B_{p,u,v}\}_{p,u,v}$, similarly to the MRA model.

\end{document}